\documentclass[conference]{IEEEtran}
\IEEEoverridecommandlockouts

\newcommand{\supp}{\mathop{\mathrm{supp}}}

\newcommand{\F}{\mathbb{F}}
\newcommand{\Fa}{\mathbb{F}_q}
\newcommand{\Fb}{\mathbb{F}_{q^{m_1}}}
\newcommand{\Fc}{\mathbb{F}_{q^{m_2}}}

\newcommand{\G}{\mathcal{G}}
\newcommand{\D}{\mathcal{D}}
\newcommand{\FF}{\mathcal{F}}
\newcommand{\C}{\mathcal{C}}
\renewcommand{\SS}{\mathcal{SS}}
\newcommand{\SC}{\mathcal{SC}}

\newcommand{\va}{\mathbf{a}}
\newcommand{\vb}{\mathbf{b}}
\newcommand{\vc}{\mathbf{c}}
\newcommand{\vA}{\mathbf{A}}

\usepackage{amsthm}
\usepackage[ruled,vlined]{algorithm2e} 
\usepackage{cite}
\usepackage{amsmath,amssymb,amsfonts}
\usepackage{algorithmic}
\usepackage{graphicx}
\usepackage{textcomp}
\usepackage{xcolor}
\def\BibTeX{{\rm B\kern-.05em{\sc i\kern-.025em b}\kern-.08em
    T\kern-.1667em\lower.7ex\hbox{E}\kern-.125emX}}

\newtheorem{theorem}{Theorem}
\newtheorem{lemma}[theorem]{Lemma}
\newtheorem{proposition}[theorem]{Proposition}
\newtheorem{corollary}[theorem]{Corollary}

\newtheorem{definition}{Definition}

\newtheorem{remark}{Remark}

\begin{document}

\title{Schur Products of Constacyclic Codes via the Constacyclic Discrete Fourier Transform
}

\author{\IEEEauthorblockN{Peifeng Lin}
\IEEEauthorblockA{\textit{Faculty of Computational Mathematics and Cybernetics} \\
\textit{Lomonosov Moscow State University}\\
Moscow, Russia \\
Email: plin@cs.msu.ru}

}

\newcommand\copyrighttext{%
   \textcopyright 2026 IEEE. Personal use of this material is permitted. Permission from IEEE must be
obtained for all other uses, in any current or future media, including
reprinting/republishing this material for advertising or promotional purposes, creating new
collective works, for resale or redistribution to servers or lists, or reuse of any copyrighted
component of this work in other works.
}
\newcommand\copyrightnotice{%
\textbf{\Large IEEE Copyright Notice}
\\
{\fbox{\parbox{\dimexpr\textwidth-\fboxsep-\fboxrule\relax}{\copyrighttext}}}
\\
\\
{\fbox{\parbox{\dimexpr\textwidth-\fboxsep-\fboxrule\relax}{Accepted to be published in: Proceedings of the IEEE ISIT 2026, June 28 - July 3, 2026, Guangzhou, China}}}

}

\copyrightnotice
\maketitle

\begin{abstract}
This paper investigates the Schur product of constacyclic codes via the constacyclic discrete Fourier transform (DFT). We first characterize key properties of the constacyclic DFT, highlighting its differences from the ordinary DFT. We then extend the concept of degenerate cyclic codes to constacyclic codes possessing a nontrivial pattern polynomial, thereby facilitating the analysis of their dimension sequences. Building on these tools, we generalize two established methods for computing the square of cyclic codes to compute the Schur product of arbitrary constacyclic codes. Finally, exploiting the inherent combinatorial structure, we derive properties of the Schur product dimension directly from additive combinatorics.
\end{abstract}

\begin{IEEEkeywords}
constacyclic codes, Schur product, constacyclic DFT.
\end{IEEEkeywords}

\section{Introduction}
Constacyclic codes, first defined by E. R. Berlekamp \cite{Berlekamp}, form a fundamental generalization of cyclic and negacyclic codes. Their algebraic structure, corresponding to ideals in the ring $\mathbb{F}[x]/(x^n - \lambda)$, enables efficient encoding and decoding via linear feedback shift registers. As a widespread class of codes encompassing well-known families like Reed–Solomon \cite{GRS} and BCH codes \cite{BCH}, constacyclic codes have been extensively studied in classical coding theory.

The Schur (Hadamard, or component-wise) product of linear codes has evolved from a theoretical construct to a critical tool in modern cryptography. Beyond its classical role in decoding, it now underpins advanced cryptographic primitives such as secure multiparty computation (MPC) and strongly multiplicative secret sharing schemes  \cite{MPC}. In these applications, the core requirement is to find codes whose dual and square both have large minimum distance—a challenging and fundamental problem in the field.

Consequently, the study of squares has been extended to various structured code families. Significant progress has been made for cyclic codes \cite{firstsumset} , affine variety codes \cite{AVcode}, and notably, BCH codes \cite{BCH2}, where both the dimensions of the codes and the bounds on the minimum distance of their squares have been actively investigated. In \cite{Falk}, the authors laid groundwork by characterizing the growth of powers of constacyclic codes, the essence of which is to study whether constacyclic codes are degenerate. However, compared to the detailed analyses available for cyclic codes \cite{firstsumset}, a comprehensive understanding of the Schur product of constacyclic codes remains incomplete. In this paper, we bridge this gap by conducting a systematic study of the Schur product for constacyclic codes. 

The paper is organized as follows. Section \MakeUppercase{\romannumeral 2} presents the necessary background on constacyclic codes, the constacyclic DFT and additive combinatorics. Section \MakeUppercase{\romannumeral 3} develops the constacyclic DFT and its key properties. Section \MakeUppercase{\romannumeral 4} analyzes the Schur product for constacyclic codes. Section \MakeUppercase{\romannumeral 5} summarizes the main results.

\section{PRELIMINARIES}
Let $q$ be a prime power, $n\in \mathbb{N}^{+}$ with $(n,q)=1$, and $\lambda \in \F_q^*$. Denote the multiplicative order of $\lambda$ in $\F_q$ by $o_\lambda$. The splitting field of $x^n-1$ is $\F_{q^{m_1}}$, where $m_1=\operatorname{ord}_n(q)$, and the splitting field of $x^n-\lambda$ is $\F_{q^{m_2}}$, where $m_2=\operatorname{ord}_{n* o_\lambda}(q)$. Let $\xi\in \Fb$ be a primitive $n$-th root of unity and let $\beta\in \Fc$ be an $n$-th root of $\lambda$.

We identify vectors $\va=(a_0,\ldots,a_{n-1})\in \F_q^n$ with polynomials $a(x)=\sum_{i=0}^{n-1}a_i x^i$. Under this identification, the support of $\va$ is $\operatorname{supp}(\va)=\{i\in \mathbb{Z}_n:a_i\neq 0\}$.

A $\lambda$-constacyclic code is an ideal in $\Fa[x]/\langle x^n-\lambda\rangle$. Hence every such code is generated by a unique monic $g(x)$ with $g(x)\mid x^n-\lambda$, and we write $\C=\langle g(x)\rangle$. Its complete generating set with respect to $\xi,\beta$ is $\G_{\xi,\beta}=\{i\in \mathbb{Z}_n:g(\xi^i\beta)\neq 0\}$.

We recall additive combinatorics on $\mathbb{Z}_n$. For $A,B\subset \mathbb{Z}_n$, the sumset is $A+B=\{a+b:a\in A,b\in B\}$, and the $t-$fold iterated sumset is $tA=A+\cdots+A$ (t times). For $A\subset \mathbb{Z}_n$, let $\langle A\rangle$ be the smallest subgroup of $\mathbb{Z}_n$ containing $A$. The smallest coset containing $A$, denoted $\SC(A)$, is $\SC(A)=a+\langle A-A\rangle$ for any $a\in A$. In $\mathbb{Z}_n$, the smallest coset containing the sumset $A+B$ satisfies $\SC(A+B)=\SC(A)+\SC(B)$.

The Schur product of $\vc_1,\vc_2\in \F_q^n$ is component-wise, and for codes $\C_1,\C_2$ we set $\C_1\circ \C_2=\operatorname{span}\{\vc_1\circ \vc_2:\vc_1\in \C_1,\vc_2\in \C_2\}$. The $t$-th Schur power is $\C^{\circ t}=\C\circ\cdots\circ\C$ (t times). We recall that the Schur product of a $\lambda_1$--constacyclic code and a $\lambda_2$--constacyclic code is  $\lambda_1\cdot \lambda_2$--constacyclic. The dimension sequence $\dim(\mathcal{C}^{\circ i})$ of a linear code $\mathcal{C}$ strictly increases until stabilizing. The stabilization exponent $r(\mathcal{C})$ is called the Castelnuovo–Mumford regularity.

Finally, let $\va \in  \Fa^n$. The  $\lambda-$constacyclic DFT of $\va$ with respect to $\xi$, $\beta$ is the vector $\vA=\FF_{\xi,\beta}(\va)\in \Fc^n$, where $A_j=\sum_{i=0}^{n-1}a_i(\xi^{j}\beta)^i$. The inverse transform of $\vA$ is given by  $a_i=\frac{1}{n \beta^i}\sum_{j=0}^{n-1}A_j \xi^{-ij}$. The spectral support of $\va$ with respect to $\xi,\beta$ is  $\SS_{\xi,\beta}(\va)=\{i\in \mathbb{Z}_n:a(\xi^i\beta)\neq 0\}$.

\section{Constacyclic Discrete Fourier Transform}
The constacyclic DFT was first introduced in \cite{CCDFT} for encoding and decoding certain cyclic codes. The constacyclic DFT differs from the ordinary DFT because of the additional factor \(\beta\).

For the ordinary DFT, if \(A_j\) is the \(j\)-th frequency component, then \(A_j^q=A_{qj}\). In contrast, for the constacyclic DFT when \(\beta\notin \mathbb{F}_q^*\), this relation fails;  instead, the Frobenius action introduces an additional shift. More precisely, if \(t\in\mathbb{Z}_n\) is defined by $\xi^t=\beta^{q-1}$, then
$A_j^q=A_{qj+t}$.

Thus the Frobenius action is no longer multiplication by \(q\) alone, but rather the affine map \(j\mapsto qj+t\). This is one of the main differences between the ordinary DFT and the constacyclic DFT.

In this paper, for fixed \(n\) and \(q\), the primitive \(n\)-th root of unity \(\xi\) is fixed. However, for a fixed \(\lambda\), different choices of the \(n\)-th root \(\beta\) are possible. We denote the corresponding \(\lambda\)-constacyclic DFT pair by
$a_i \overset{\beta}{\longleftrightarrow} A_j$, omitting \(\xi\) for simplicity and emphasizing the dependence on \(\beta\).

The work in \cite{Wrongarxiv} assumed that \(\lambda\) has an \(n\)-th root in \(\mathbb{F}_q^*\). In that case, \(\beta\in\mathbb{F}_q\), so \(\beta^{q-1}=1\), implying \(t=0\). Hence the Frobenius action reduces to \(j\mapsto qj\), making the behavior appear identical to that of the ordinary DFT. In this paper, we remove this restriction. The following lemma establishes the existence and uniqueness of such a $t$.
\begin{lemma}
\label{lemma:existence-of-t}
	For any primitive $n-$th root of  unity $\xi\in  \Fb$ and  $n-$th root $ \beta \in \Fc$ of $\lambda\in \Fa$, there exists a unique  $t\in \mathbb{Z}_n$ such that $\xi^t=\beta^{q-1}$. 
\end{lemma}
\begin{proof}
	Let  $\delta$ be a primitive $n\cdot o_\lambda$-th root of  unity. We can write $\xi=\delta^{i\cdot o_\lambda}$, $\beta=\delta^{j}\xi^l$, where $i\in [1,n-1]$ is fixed by $\xi$  with $(i,n)=1$, $j\in [1,n-1]$ is fixed by  $\lambda$ with $(j,o_\lambda)=1$, and $l\in \mathbb{Z}_n$.
 Expressing $\beta^{q-1}$ and $\xi^t$ as powers of $\delta$:
$$
\beta^{q-1}=\delta^{j(q-1)+i\cdot o_\lambda\cdot l (q-1)}, \xi^t=\delta^{i\cdot o_\lambda \cdot t}.
$$
Equality holds if and only if 
$$
i\cdot o_\lambda \cdot t \equiv j(q-1)+i\cdot o_\lambda\cdot l (q-1) \pmod{n\cdot o_\lambda},
$$
 which, since $o_\lambda|(q-1)$ and $(i,n)=1$, is equivalent to 
 $$t \equiv (o_\lambda\cdot l + j\cdot i^{-1}\pmod{n})\frac{q-1}{o_\lambda} \pmod{n} .$$

\end{proof}
\begin{proposition}
	\label{proposition:closure1}
	 Let $\vA \in \F^{n}_{q^{m_2}}$, $\va = \mathcal{F}^{-1}_{\xi,\beta}(\vA)$ and $\xi^t=\beta^{q-1}$. Then $\va\in \F^n_{q}$ if and only if  $ A^q_j = A_{qj+t},\forall j\in \mathbb{Z}_n$. 
\end{proposition}

\begin{proof}
Suppose that	 $\va\in \F_q^n$. Then $a_i^q=a_i,\forall i \in \mathbb{Z}_n$. We have
$$
A_j^q=\left(\sum_{i=0}^{n-1}a_i(\xi^{j} \beta)^i\right)^q=\sum_{i=0}^{n-1}a_i^q(\xi^{jq} \beta^{q})^i=
$$
$$
=\sum_{i=0}^{n-1}a_i(\xi^{jq} \beta^{1+(q-1)})^i=\sum_{i=0}^{n-1}a_i(\xi^{qj+t} \beta)^i=A_{qj+t}.
$$
Conversely, suppose $A^q_j = A_{qj+t}, \forall j \in \mathbb{Z}_n$. Then 
$$0=A^q_j - A_{qj+t}=\sum_{i=0}^{n-1}(a_i^q-a_i)(\xi^{qj+t}\beta)^i, \forall j \in \mathbb{Z}_n$$
Applying the inverse constacyclic DFT then yields $a_i^q=a_i,\forall i\in \mathbb{Z}_n$, which means $\va\in \F^n_q$.
\end{proof}

\begin{proposition}
\label{proposition:closure2}
	Let $K \subset \mathbb{Z}_n$ and  $\xi^t=\beta^{q-1}$. Then  polynomial $g(x)=\prod_{k\in K}(x-\xi^k\beta)$ has coefficients in $\Fa$ and  divides $x^n-\beta^n$   if and only if $k \in K$ implies $qk+t \in K$.
\end{proposition}

\begin{proof}
	The proof is similar to that of Lemma \ref{lemma:existence-of-t}. Expressing $\xi^k\beta$ as powers of $\delta$:
 $$\xi^k\beta=\delta^{(o_\lambda i)\cdot k+(i\cdot l\cdot o_\lambda+j)}.$$ Observe that $$\xi^{kq+t}\beta = \xi^{kq}\beta^{q}=(\xi^k\beta)^q.$$
Thus, closure of $K$ under the map $k \mapsto kq+t \pmod{n}$ implies closure under the Frobenius map $k \mapsto kq \pmod{n\cdot o_{\beta^n}}$ , which corresponds to the condition for polynomials over $\Fa$ in the ordinary DFT setting. 
\end{proof}

The following proposition is key for relating the complete generating set of the Schur product of two constacyclic codes to the sumset of their complete generating sets.
\begin{proposition}
	\label{proposition:componentwise-conv}
	Let $a_i \overset{\beta_1}{\longleftrightarrow}   A_j $ and $b_i \overset{\beta_2}{\longleftrightarrow}   B_j $. Then
$$
a_i \cdot b_i \quad \overset{\beta_1\cdot \beta_2}{\longleftrightarrow} \quad \frac{1}{n} \sum_{t=0}^{n-1} A_t\cdot B_{j-t}.
$$
\end{proposition}

\begin{proof}
	Let $\vc = \va \circ \vb$ and $c_i \overset{\beta_1\beta_2}{\longleftrightarrow}   C_j $. We shall prove that $C_j= \frac{1}{n} \sum_{t=0}^{n-1} A_t \cdot B_{j-t}$.
 $$
  \frac{1}{n} \sum_{t=0}^{n-1} A_t \cdot B_{j-t} = \frac{1}{n} \sum_{t=0}^{n-1} \sum_{r=0}^{n-1} \sum_{s=0}^{n-1}a_rb_s \xi^{tr+(j-t)s} \beta_1^{r}\beta_2^s = 
 $$
 \begin{equation}\label{equ:sum-n-or-zero}
     = \frac{1}{n} \sum_{r=0}^{n-1} \sum_{s=0}^{n-1}  a_rb_s \xi^{js} \beta_1^{r}\beta_2^s \left(\sum_{t=0}^{n-1} \xi^{t(r-s)}\right)
 \end{equation}

 Since $ \sum_{t=0}^{n-1} \xi^{t(r-s)} = n\delta_{r,s}$, where $\delta_{r,s}$ is the Kronecker delta, equation (\ref{equ:sum-n-or-zero}) reduces to 
 $$\frac{1}{n}\sum_{i=0}^{n-1}a_ib_i (\xi^j\beta_1\beta_2)^i n= C_j.\qedhere$$
\end{proof}

\begin{proposition} \cite{CCDFT}
\label{DFT:conv-component}
Let $a_i \overset{\beta}{\longleftrightarrow}   A_j $ , $b_i \overset{\beta}{\longleftrightarrow}   B_j $.
If $c(x)=a(x)*b(x)  \pmod {x^n-\beta^n}$, then 
$$
c_i \quad \overset{\beta}{\longleftrightarrow} \quad A_j \cdot B_{j}.
$$
\end{proposition}

To describe the complete generating set of the dual of a constacyclic code, we first need the following reversal property of the constacyclic DFT. 

\begin{proposition}
\label{DFT:reverse}
Let $\va,\vb \in \F_q^n$. Suppose
	\[
	a_i \overset{\beta}{\longleftrightarrow} A_j,
	\qquad
	b_i \overset{\beta^{-1}}{\longleftrightarrow} B_j,
	\qquad
	b_i=a_{n-1-i}, \quad \forall i\in \mathbb{Z}_n.
	\]
	Then
	\[
	B_j=A_{n-j}\cdot \beta^{1-n}\xi^{-j},
	\qquad \forall j\in \mathbb{Z}_n.
	\]
\end{proposition}

\begin{proof}
\begin{align*}
   &B_j = \sum_{i=0}^{n-1}b_i (\xi^j\beta^{-1})^i=
    \quad\text{(with } t = n - 1 - i\text{)} \\
    &=\sum_{t=0}^{n-1}b_{n-1-t} (\xi^j\beta^{-1})^{n-1-t} = \sum_{t=0}^{n-1}a_{t} (\xi^j\beta^{-1})^{n-1-t}=\\
    &= \sum_{t=0}^{n-1}a_{t} (\xi^{-j}\beta^{})^{t} \xi^{(n-1)j}\beta^{1-n}=\sum_{t=0}^{n-1}a_{t} (\xi^{n-j}\beta^{})^{t} \xi^{-j}\beta^{1-n}=\\
    &=A_{n-j} \cdot \beta^{1-n}\cdot \xi^{-j}. 
\end{align*}
\end{proof}
\begin{remark}
   The statement in \cite[Theorem 5]{Wrongarxiv} considers a different reversal operation for vectors, while both transforms are taken with respect to  $\xi, \beta$.

We point out that the last equality in equation (5) in \cite{Wrongarxiv} does not appear to hold. This is because, under this setting, the equality fails when both spectrum vectors are defined with respect to $\xi, \beta$. A consistent formulation instead requires using $\xi, \beta$ for one vector and $\xi, \beta^{-1}$ for the other.
\end{remark}

\begin{corollary}
Let $\C$ be a constacyclic code with complete generating set $\G_{\xi,\beta}$. Then the complete generating set of dual code $\C^\perp$ with respect to $\xi$, $\beta^{-1}$ is given by
$-(\mathbb{Z}_n\setminus \G_{\xi,\beta})$.
\end{corollary}

\begin{proof}
    It follows from that the dual code of the $\lambda-$constacyclic code $\langle g\rangle$ is generated by $(\frac{x^n-\lambda}{g})^*$, where $*$ denotes the reciprocal polynomial.  Given the above generator, this conclusion is an immediate consequence of Propositions \ref{DFT:conv-component} and \ref{DFT:reverse}.
\end{proof}

\begin{proposition}
	\label{proposition:xv}
Let $\C$ be a constacyclic code of length $n$ with complete generating set $\G$. Then $\G$ is a union of cosets of a nontrivial subgroup in $\mathbb{Z}_n$ if and only if $\C$ is generated by a polynomial $g(x)$ that can be expressed as  $g(x)=u(x^v)$, where $v\mid n$ and $1< v < n$.	
\end{proposition}

\begin{proof}
 Let $\D=\mathbb{Z}_n \setminus \G$.\\
($\Rightarrow$) Suppose $\D =\cup_{i=1}^{s}D_i =\cup_{i=1}^{s}a_i + \langle \frac{n}{v}\rangle$.   Then $g(x)=\prod_{i=1}^s g_i(x)$ with $g_i(x)=\prod_{j\in D_i} (x-\xi^j\beta)$.
Let $y = \frac{x}{\xi^{a_i}\beta}$. Then 
$g_i(x)=(\xi^{a_i}\beta)^v \prod_{j=0}^{v-1}(y-\xi^{j\cdot \frac{n}{v}}).$
Since $\xi^\frac{n}{v}$ is a primitive $v$-th  root of unity, $\prod_{j=0}^{v-1}(y-\xi^{j\cdot \frac{n}{v}})=y^v-1$. Thus 
$$g_i(x)=x^v-(\xi^{a_i}\beta)^v=u_i(x^v).$$
Hence, 
$$g(x)=\prod_{i=1}^s g_i(x)=\prod_{i=1}^s u_i(x^v)=u(x^v).$$

($\Leftarrow$) Suppose $g(x)=u(x^v)$.  If  $g(\xi^i \beta)=0$, then $u(\xi^{iv} \beta^v)=0$.  
Since $\xi^{iv} = \xi^{(i+n/v)v}$, we have
$$g(\xi^{i+n/v} \beta) = u(\xi^{(i+n/v)v} \beta^v) = u(\xi^{iv} \beta^v)=0.$$ 
Thus, $\mathcal{D}$ is closed under addition by $\frac{n}{v}$, meaning $\G$ is a union of cosets of $\langle \frac{n}{v} \rangle$.
\end{proof}

\section{Properties of Constacyclic Codes Under the Schur Product}
\subsection{Degenerate constacyclic codes}

In \cite{MacWilliams}, a cyclic code $\C=\langle g\rangle$ is called degenerate if $g(x)$ is divisible by a polynomial of the form  $\sum_{i=0}^{\frac{n}{v}-1}x^{i\cdot v}$ .

We generalize this concept to constacyclic codes, which helps determine whether powers of a constacyclic code can generate the full space $\F^n$.

In \cite{Falk} pattern polynomial $p(x)$ of the generator polynomial $g(x)$ for a  constacyclic code $\C$ of length $n$ is defined as a polynomial with minimal $v$ dividing $n$ and $\alpha\in \F_q^*$, such that $p(x)=\sum_{i=0}^{\frac{n}{v}-1}\alpha^i x^{i\cdot v}$  divides $ g(x)$. 
\begin{definition}
Constacyclic code $\C=\langle g(x) \rangle$ of length $n$ with nontrivial pattern polynomial $p(x)=\sum_{i=0}^{\frac{n}{v}-1}\alpha^ix^{iv} \neq 1 $, where $p(x)|g(x)$, $v\mid n$, $v<n$ is called degenerate. If the pattern polynomial is trivial, i.e., $p(x)=1$, the code is non-degenerate.
\end{definition}

\begin{proposition} \cite{Falk}
	For constacyclic code $\C$ over finite field $\F$, there exist an integer  $i>1$ such that $\C^{\circ i} = \F^n$ if and only if $\C$ is non-degenerate.
\end{proposition}

\begin{proposition}
\label{pro:smallest-coset}
Let $\C=\langle g(x)\rangle$ be a $\lambda$-constacyclic code with pattern polynomial $p(x)$. 
The support $\operatorname{supp}(\mathcal{F}_{\xi,\beta}(p(x)))$ of the DFT of $p(x)$ is the smallest coset of a subgroup of $\mathbb{Z}_n$ that contains the complete generating set $\mathcal{G}$ of the code, i.e.,
$$\SC(g(x))=\SS(p(x)).$$ 
\end{proposition}
\begin{proof}
Let $g(x)=p(x) \cdot u(x) \pmod{x^n-\lambda}$, where $p(x)=\sum_{i=0}^{\frac{n}{v}-1}\alpha^ix^{vi}$. $p(\xi^k \beta) \neq 0$ if and only if $\alpha(\xi^k \beta)^v -1= 0$, i.e., $\xi^{kv} = \alpha^{-1} \beta^{-v}$. This condition defines a coset of the subgroup $\langle n/v \rangle$ in $\mathbb{Z}_n$. By Proposition \ref{DFT:conv-component}, 
$$\FF(g(x))=\FF(p(x))\circ \FF(u(x)).$$
Consequently, $\G= \supp(\FF(g(x)))\subseteq \supp(\FF(p(x))).$

The minimality of $v$ in the definition of the pattern polynomial implies that the subgroup $\langle n/v \rangle$ (and hence its coset) is the smallest possible with this property.
\end{proof}

To find the pattern polynomial of a constacyclic code $\langle g(x)\rangle$, we need to determine $v$ and $\alpha$. As shown in \cite{Falk}, if $v$ is known, $\alpha$ can be obtained by dividing the coefficient $g_v$ by $g_0$ (provided $g_0 \neq 0$). In practice, for a given complete generating set $\G \subset \mathbb{Z}_n$, we can find the minimal coset containing $\G$ by translating $\G$ to contain $0$ and then computing $d = \gcd(\G',n)$, where $\G' = \G - \min(\G)$. Then $v = n/d$. This method is summarized in the following algorithm.
\begin{algorithm}
\label{alg:pattern-poly}
\caption{Computation of the pattern polynomial of a constacyclic code} 
\label{alg:123} 
\KwIn{$\mathbf{g}\in \Fa^n$, $\xi, \beta$.} 
\KwOut{$v$, $\alpha$.} 
$A \gets \supp(\FF_{\xi,\beta}(\mathbf{g}))$\; $t \gets \gcd(A-\min A,n)$\;
\If{$t = 1$}{ 
    \Return{$n,1$}\;
}
\Else{ 
    \Return{$\frac{n}{t}, g_{v}\cdot g_0^{-1}$}\;
}
\end{algorithm}

A degenerate constacyclic code $\C_1$ of length $n$ with pattern polynomial $p(x) \neq 1$ is in one-to-one correspondence with a non-degenerate constacyclic code $\C_2$ of length $v$ via the mapping $\C_1 \to \C_2: c(x) \mapsto \frac{c(x)}{p(x)} $. We denote this non-degenerate code code by $\overline{\C}$ and call it the core of $\C_1$.

In \cite{Falk}, it is shown that for constacyclic code $\C$ with nontrivial pattern polynomial, the generator polynomial of power $\C^{\circ i\ge r(\C)}$  is given by  $(p(x))^{\circ i}$. This follows from the one-to-one correspondence between $\C^{\circ i}$ and $(\overline{\C})^{\circ i}$.

\begin{theorem}
	For a constacyclic code $\C$ with nontrivial pattern polynomial $p(x)$, the generator polynomial of $\C^{\circ i}$ is given by $\overline{g}_i(x)(p(x))^{\circ i}$, where $\overline{g}_i(x)$ is the generator polynomial of $(\overline{\C})^{\circ i}$.
\end{theorem}
\begin{proof}
	Let $p(x)=\sum_{i=0}^{\frac{n}{v}-1}\alpha^i x^{iv}$. Then any codeword in $\C$ can be written as 
 \[
 (\vc \mid \alpha \vc \mid \dots \mid \alpha^{\frac{n}{v}-1}\vc),
 \]where $ \vc \in  \overline{\C}$. Consequently, any codeword in  $\C^{\circ i}$ has the form  
 \[
 (\vc \mid \alpha^i \vc \mid \dots \mid \alpha^{i(\frac{n}{v}-1)}\vc),
 \]where $ \vc \in  (\overline{\C})^{\circ i}$. Thus,  $\C^{\circ i}$ is generated by $\overline{g}_i(x)(p(x))^{\circ i}$.
\end{proof}

\subsection{Schur Product of Constacyclic Codes via constacyclic DFT}
In \cite[Theorem 1]{firstsumset}, it is shown that for a cyclic code \(\C\), the complete defining set of its Schur square \(\C^{\circ 2}\) is the 2-fold iterated sumset of the complete defining set of \(\C\). Our goal is to generalize this viewpoint from the ordinary DFT to the constacyclic DFT.

To do so, we follow the approach of \cite{firstsumset} and extend constacyclic codes from the base field \(\mathbb{F}_q\) to the splitting field \(\mathbb{F}_{q^{m_2}}\) before applying the constacyclic DFT. The reason is that codewords in the time domain are vectors over \(\mathbb{F}_q\), whereas their constacyclic DFTs lie in \(\mathbb{F}_{q^{m_2}}\). Passing to the larger field allows convolution arguments to be carried out directly in the frequency domain without changing the underlying code structure.

Extending a code \(\mathcal{C}\subseteq \mathbb{F}_q^n\) by scalars from \(\mathbb{F}_{q^{m_2}}\) yields the code $\mathbb{F}_{q^{m_2}}\otimes \mathcal{C}$, consisting of all \(\mathbb{F}_{q^{m_2}}\)-linear combinations of codewords in \(\mathcal{C}\).

For a subset \(S\subseteq \mathbb{Z}_n\), let \(V_S\) denote the subspace of \(\mathbb{F}_{q^{m_2}}^n\) consisting of all vectors whose support is contained in \(S\), namely
\[
V_S=\{\mathbf{c}\in \mathbb{F}_{q^{m_2}}^n:c_i=0\text{ for all }i\notin S\}.
\]
\begin{lemma}\label{lemma:image-of-dft}
	For a constacyclic code $\C$ with complete generating set $\G$, the image of its extension $\FF(\Fc \otimes \C)$  under the constacyclic DFT is $V_\G$.
\end{lemma}
\begin{proof}
	Let $g(x)$ be the generator polynomial of $\C$. Then the image of  $\Fc \otimes \C$ is spanned by  
 \begin{equation*} \label{equation:Fourier-image}
     \{\FF(x^ig(x)): 0\le i <  n-\deg(g) \}.
 \end{equation*}
 The $j-$th component of $\FF(x^ig(x))$ is nonzero if and only if $j\in \G$. Given that, these vectors are linearly independent, hence the image is precisely $V_\G$. 
\end{proof}

\begin{theorem}
\label{theorem:sumset}
	Let $\C_i$ be $\beta_i^n$--constacyclic codes with complete generating sets $\G_i$ with respect to $\xi$ and $\beta_i$ for $i=1,2$. Then the complete generating set of  $\C_1 \circ \C_2$ with respect to $\xi$ and $\beta_1\beta_2$ is  $\G_1+\G_2$.
\end{theorem}
\begin{proof}
	Let \(\G_i\) be the complete generating set of \(\C_i\) with respect to \((\xi,\beta_i)\), and let \(t_i\in\mathbb{Z}_n\) satisfy \(\xi^{t_i}=\beta_i^{q-1}\). By Proposition  \ref{proposition:closure2}, each \(\G_i\) is closed under \(j\mapsto qj+t_i\). Hence \(\G_1+\G_2\) is closed under \(j\mapsto qj+(t_1+t_2)\), and so is its complement.

Therefore, by Proposition \ref{proposition:closure2}
\[
g_3(x)=\prod_{i\in \mathbb{Z}_n\setminus(\G_1+\G_2)}(x-\xi^i\beta_1\beta_2)
\]
has coefficients in \(\mathbb{F}_q\) and divides \(x^n-(\beta_1\beta_2)^n\). Let \(\C_3=\langle g_3(x)\rangle\). By Lemma \ref{lemma:image-of-dft},
\[
\mathcal{F}_{\xi,\beta_1\beta_2}(\mathbb{F}_{q^{m_2}}\otimes \C_3)=V_{\G_1+\G_2}.
\]

On the other hand, scalar extension commutes with the Schur product:
\[
\mathbb{F}_{q^{m_2}}\otimes(\C_1\circ \C_2)
=
(\mathbb{F}_{q^{m_2}}\otimes \C_1)\circ(\mathbb{F}_{q^{m_2}}\otimes \C_2).
\]
By Proposition \ref{proposition:componentwise-conv} and Lemma \ref{lemma:image-of-dft},
\[
\mathcal{F}_{\xi,\beta_1\beta_2}(\mathbb{F}_{q^{m_2}}\otimes(\C_1\circ \C_2))
=
\operatorname{Conv}(V_{\G_1},V_{\G_2})
=
V_{\G_1+\G_2}.
\]

Thus \(\C_1\circ \C_2\) and \(\C_3\) have the same DFT image after scalar extension, hence they are equal. Therefore the complete generating set of \(\C_1\circ \C_2\) is \(\G_1+\G_2\).
\end{proof}

In \cite[Proposition 4.2.2]{gcd}, the Schur square of a cyclic code is computed by fixing a generator polynomial \(g(x)\), taking Schur products of \(g(x)\) with its shifts \(x^jg(x)\), and then taking the gcd of all such products. Using Theorem 13, we extend this method from Schur squares of cyclic codes to Schur products of arbitrary constacyclic codes.
\begin{theorem}
\label{theorem:gcd}

Let $\C_i = \langle g_i(x) \rangle$ be $\lambda_i$-constacyclic codes of length $n$, $i=1,2$, and let $s(x)$ be a polynomial such that $\gcd(s(x), (x^n-\lambda_1)/g_1(x)) = 1$. Then the Schur product $\C_1 \circ \C_2$ is generated by the polynomial
$\gcd(P,\ x^n-\lambda_1\lambda_2),$
where
$$P=\{g\circ x^jg_2:0\le j<n-\deg(g_2)\}\text{ and } g(x)=g_1(x)s(x).$$
\end{theorem}
\begin{proof}
Since \(g\in \C_1\) and \(x^jg_2\in \C_2\), by definition of the Schur product of codes,
\[
g\circ x^jg_2\in\C_1\circ \C_2.
\]
Hence every polynomial in
\[
P=\{g\circ x^jg_2:0\le j<n-\deg(g_2)\}
\]
is a multiple of the generator polynomial \(g_3(x)\) of \(\C_1\circ \C_2\).
Therefore
\[
g_3(x)\mid \gcd(P).
\]

By Lemma \ref{lemma:image-of-dft}, the \(\mathbb{F}_{q^{m_2}}\)-span of
$\{\mathcal{F}_{\xi,\beta_2}(x^jg_2):0\le j<n-\deg(g_2)\}$
is precisely \(V_{\G_2}\). Since \(\mathcal{F}_{\xi,\beta_1}(g)\) has support \(\G_1\), Proposition \ref{proposition:componentwise-conv} implies that the \(\mathbb{F}_{q^{m_2}}\)-span of the constacyclic DFTs of the polynomials in \(P\) is exactly \(V_{\G_1+\G_2}\).

Hence the common spectral zero set of the polynomials in \(P\) is precisely
$
\mathbb{Z}_n\setminus(\G_1+\G_2)$.

By Theorem 13, this is also the spectral zero set of \(g_3(x)\).

However, the constacyclic DFT only detects roots among the roots of \(x^n-\lambda_1\lambda_2\). Thus we can only conclude that
\[
\gcd(P)=g_3(x)u(x),
\]
where \(u(x)\) is coprime to \(x^n-\lambda_1\lambda_2\). Therefore
\[
g_3(x)=\gcd(\gcd(P),x^n-\lambda_1\lambda_2). \qedhere
\]
\end{proof}
According to \cite[Theorem 2]{Falk}, for a constacyclic code \(\C\), the pattern polynomial of \(\C^{\circ t}\) coincides with the \(t\)-th Schur power of the pattern polynomial of \(\C\). We now extend this principle from Schur powers of a single code to the Schur product of two arbitrary constacyclic codes.

\begin{theorem}
	For $\lambda_i$-constacyclic codes $\C_i$ with pattern polynomials $p_i(x)$, $i=1,2$, the pattern polynomial of $\C_1 \circ \C_2$ is  $p_1(x) \circ p_2(x)$.
\end{theorem}
\begin{proof}
    If either code is non-degenerate, then, by commutativity and associativity,
\(
(\C_1\circ \C_2)^{\circ n}=\C_1^{\circ n}\circ \C_2^{\circ n}=\mathbb{F}_q^n
\).
Hence the pattern polynomial of \(\C_1\circ \C_2\) is trivial. Therefore, we may assume that both codes are degenerate.

Let \(\G_i\) be the complete generating set of \(\C_i\), $i=1,2$, and $k_2=\dim(\C_2)$.

Write
\(
p_i(x)=\sum_{t=0}^{n/v_i-1}\alpha_i^t x^{v_i t}
\),
let \(v_3=\operatorname{lcm}(v_1,v_2)\), and define
\(
\alpha_3=\alpha_1^{v_3/v_1}\alpha_2^{v_3/v_2}
\).
Then
\[
p_3(x):=\sum_{t=0}^{n/v_3-1}\alpha_3^t x^{v_3 t}
      =p_1(x)\circ p_2(x)\mid x^n-\lambda_1\lambda_2.
\]

We first show that \(p_3(x)\mid g_3(x)\), where \(g_3(x)\) is the generator polynomial of \(\C_1\circ \C_2\). By Theorem 14,
\[
g_3(x)=\gcd(g_1\circ g_2,\ g_1\circ xg_2,\ \dots,\ g_1\circ x^{k_2-1}g_2, x^n-\lambda_1\lambda_2).
\]
Thus, it suffices to prove that
\(
p_3(x)\mid (g_1(x)\circ x^j g_2(x))
\)
for every \(0\le j<k_2\).

Since \(\SS_{\xi,\beta_2}(x^j)=\mathbb{Z}_n\), multiplication by \(x^j\) does not affect the spectral support. Hence \(x^j g_2(x)\) has the same complete generating set as \(g_2(x)\). By Propositions 4 and 5,
\[
\SS_{\xi,\beta_1\beta_2}(g_1(x)\circ x^j g_2(x))
\subseteq \SS_{\xi,\beta_1}(p_1(x))+\SS_{\xi,\beta_2}(p_2(x)).
\]
Consequently, every element outside
\(
\SS_{\xi,\beta_1}(p_1(x))+\SS_{\xi,\beta_2}(p_2(x))
\)
is a spectral zero of \(g_1(x)\circ x^j g_2(x)\).

Since \(\SS_{\xi,\beta_1}(p_1(x))\) and \(\SS_{\xi,\beta_2}(p_2(x))\) are cosets with step sizes \(v_1\) and \(v_2\), respectively, their sumset is a coset with step size \(v_3=\operatorname{lcm}(v_1,v_2)\). Hence, by Proposition \ref{proposition:componentwise-conv}
\begin{equation}\label{eq:support-sum}
\SS_{\xi,\beta_1}(p_1(x))+\SS_{\xi,\beta_2}(p_2(x))
=\SS_{\xi,\beta_1\beta_2}(p_3(x)).
\end{equation}
Therefore, every spectral zero of \(p_3(x)\) is also a spectral zero of \(g_1(x)\circ x^j g_2(x)\), and thus
\(
p_3(x)\mid (g_1(x)\circ x^j g_2(x))
\).
Since this holds for all \(j\), we conclude that \(p_3(x)\mid g_3(x)\).

Finally, by Theorem 13, the complete generating set of \(\C_1\circ \C_2\) is
\(
\G_3=\G_1+\G_2
\).

By Proposition \ref{pro:smallest-coset}, we obtain
\[
\SC(\G_3)=\SC(\G_1)+\SC(\G_2)
=\SS_{\xi,\beta_1}(p_1(x))+\SS_{\xi,\beta_2}(p_2(x)).
\]
Together with~\eqref{eq:support-sum}, this shows that
\(p_3(x)\) is the unique minimal pattern polynomial whose spectral support contains \(\G_3\). Hence \(p_3(x)\) is the pattern polynomial of \(\C_1\circ \C_2\), and therefore
\[
p(\C_1\circ \C_2)=p_1(x)\circ p_2(x). \qedhere
\]

\end{proof}

Given that constacyclic codes have inherent combinatorial properties, several characteristics of their Schur product dimension follow from additive combinatorics.
\begin{proposition}
	Let $\C_1$, $\C_2$ be constacyclic codes of dimensions $k_1$, $k_2$ respectively. If $k_1+k_2 > n$, then $\C_1 \circ \C_2 = \Fa^n$.
\end{proposition}
\begin{proof}
	See \cite[Exercise 2.1.6]{AddiComb} 
\end{proof}

\begin{proposition}
 Let $\C$ be an $[n,k \ge 2]$ constacyclic code with pattern polynomial $\sum_{i=0}^{\frac{n}{v}-1}\alpha^ix^{vi}$. Then
 \[
 r(\C) \leq 
\begin{cases}
\frac{n-1}{k-1}, & \text{if } n \text{ is prime} \\
\frac{2v}{k}, & \text{if } n \text{ is composite}
\end{cases}
 \]
\end{proposition}
\begin{proof}
	See \cite[Exercises 12.1.1, 12.1.2]{AddiComb}. 
 The complete generating set $\G$ can be translated to contain $0$ by changing $\beta$, allowing the application of additive combinatorial bounds on the order of subsets of $\mathbb{Z}_n$. 
\end{proof}

For a non-degenerate constacyclic code $\C$ with complete generating set $\G \subseteq \mathbb{Z}_n$, the Fourier bias $||\G||_u $ provides a sufficient criterion for the iterated sumset $s\mathcal{G}$ to cover the whole group $ \mathbb{Z}_n$ (equivalently, $\C^{\circ s}=\F^n_q$). As defined in \cite{AddiComb}, $||\G||_u $ is the supremum over all non-zero frequencies of the absolute Fourier coefficients of indicator function $\hat{1_{\G}}$.
\begin{proposition}
	Let $\C$ be a $[n,k]$ non-degenerate constacyclic code with complete generating set $\G$ with Fourier bias $||\G||_u$. Then 
 $$\forall s \geq \mathop{\mathrm{max}}(3, \frac{2log||\G||_u - log\frac{n}{k}}{log||\G||_u - log\frac{n}{k}}), \hspace{3pt}  \C^{\circ s}=\F^n_q.$$
\end{proposition}
\begin{proof}
	See \cite[Lemma 4.13]{AddiComb}.
\end{proof}

\section{Conclusion}
This paper provides a systematic framework for studying the Schur product of constacyclic codes. We extend the concept of degenerate cyclic codes to constacyclic codes via pattern polynomials and characterized their structure using the constacyclic DFT. Building on this foundation, we generalize two key methods—sumset and GCD-based—for computing the Schur product from cyclic to constacyclic codes. The combinatorial nature of constacyclic codes further allows us to derive bounds on their Castelnuovo–Mumford regularity and conditions for Schur powers to fill the space, leveraging results from additive combinatorics. This work bridges the gap between cyclic and constacyclic code products, offering tools for applications in coding theory and cryptography. 
\newpage

\bibliographystyle{IEEEtran} 
\bibliography{references}

\end{document}